\documentclass{llncs}
\usepackage{llncsdoc}
\usepackage{caption}
\usepackage{subcaption}
\usepackage{graphicx}
\usepackage{comment}
\usepackage{verbatim}
\usepackage{amssymb}
\usepackage[ruled,vlined,linesnumbered]{algorithm2e}
\begin{document}

\frenchspacing

\title{\Large Speeding up Graph Algorithms\\ via Switching Classes}
\institute{Colorado State University, Fort Collins CO 80521, USA}
\author{Nathan Lindzey~\thanks{lindzey@math.colostate.edu,
Mathematics Department,
Colorado State University,
Fort Collins, CO, 80523-1873
U.S.A.}
}
\date{}
\maketitle

\begin{abstract} \small\baselineskip=9pt 
Given a graph $G$, a \emph{vertex switch} of $v \in V(G)$ results in a new graph where neighbors of $v$ become nonneighbors and vice versa.  This operation gives rise to an equivalence relation over the set of labeled digraphs on $n$ vertices.  The equivalence class of $G$ with respect to the switching operation is commonly referred to as $G$'s \emph{switching class}. The algebraic and combinatorial properties of switching classes have been studied in depth; however, they have not been studied as thoroughly from an algorithmic point of view.  The intent of this work is to further investigate the algorithmic properties of switching classes. In particular, we show that switching classes can be used to asymptotically speed up several super-linear unweighted graph algorithms. The current techniques for speeding up graph algorithms are all somewhat involved insofar that they employ sophisticated pre-processing, data-structures, or use ``word tricks" on the RAM model to achieve at most a $O(\log(n))$ speed up for sufficiently dense graphs.  Our methods are much simpler and can result in super-polylogarithmic speedups. In particular, we achieve better bounds for diameter, transitive closure, bipartite maximum matching, and general maximum matching.
\end{abstract}

\section{Introduction}

The runtime of an algorithm is intimately related to how an instance is represented.  
Recall that the runtimes of the first generation of graph algorithms were expressed solely in terms of $n$, the number of vertices.  This analysis was natural since at this time graphs were represented in $\Theta(n^2)$ space via their adjacency matrix.  It was soon noticed that if $m = o(n^2)$, then a variety of graph algorithms could be sped up by first computing the adjacency list from the adjacency matrix, then running the algorithm on the more efficient adjacency list representation. This motivated the introduction of $m$ to the runtime of graph algorithms and it is now customary in algorithm design to assume that a graph instance is given in the form of its adjacency list. 

We introduce $\widetilde{m}$ as a measure of complexity and show many classical graph algorithms can be analyzed in terms of $\widetilde{m}$.  This is a significant measure of complexity since $\widetilde{m} = O(m)$ but $\widetilde{m} \neq \Theta(m)$.  In particular, if $\widetilde{m} =  o(m)$, then several graph algorithms can be asymptotically sped up by computing the so-called \emph{partially complemented adjacency list}  (pc-list) from an adjacency list, then running the algorithm on the more efficient partially complemented adjacency list representation.  

The pc-list~\cite{DahlhausGM02} is a natural generalization of the adjacency list  that involves an additional $O(n)$ bits of storage to represent \emph{vertex switches}.  When a vertex is switched, its neighbors become nonneighbors and nonneighbors become neighbors.   A (di)graph afforded such a switching operation is commonly referred to as a \emph{switching class}~\cite{Seidel76}.  Figure 1 of the appendix demonstrates how a pc-list can represent a switching class which can in turn be used to obtain a more compact representation of a graph.  Algebraic and combinatorial properties of switching classes have been studied in depth~\cite{Seidel76,ChengW86}; however, they have not been studied as thoroughly from an algorithmic point of view.  The intent of this work is extend~\cite{DahlhausGM02} by further investigating algorithmic properties of switching classes.

In~\cite{DahlhausGM02} canonical $\Theta(n+m)$ unweighted graph algorithms were developed for switching classes; however, due to the linear-time solvability of these problems, the pc-list provided no asymptotic speed up in runtime.  We extend this work by developing switching class algorithms for classical unweighted graph problems for which no linear-time algorithm is known.  We show that for sufficiently dense graphs, the pc-list can provide super-polylogarithmic speed ups in runtime. 

This is notable since the current techniques for speeding up algorithms over dense instances are all somewhat involved and achieve at most a $O(\log(n))$ speed up.  A data-structure in~\cite{KaoOT98} is given that allows one to work on the complement of a graph without constructing it; however, the algorithms they consider are linear and do not improve any of the results established in~\cite{DahlhausGM02}.
The techniques in~\cite{CheriyanM96} are notable in that they achieve a $O(\log(n))$ speed-up for several canonical graph problems over arbitrary dense graphs (assuming the RAM model).  Clever but complicated preprocessing in~\cite{FederM95} allows for an asymptotic speedup that is logarithmic in the density of the graph that is at most $O(\log(n))$.  

Our approach is much simpler insofar that it involves only basic preprocessing of the graph and slight modifications to existing algorithms.

\section{Preliminaries}
All graphs are assumed to be finite, labeled, directed, unweighted, and simple unless stated otherwise, and let $\mathcal{G}$ denote the class of all such graphs on $n$ vertices. Let $V(G)$ and $E(G)$ denote vertex set and edge set of $G$ respectively.  Let $E(A,B)$ denote the set of edges that have exactly one endpoint in $A \subseteq V$ and exactly one endpoint in $B \subseteq V$.  Let $G[X]$ denote the subgraph induced by the vertex set $X \subseteq V$.  An \emph{out-switch} (\emph{in-switch}) of a vertex changes out-neighbors (in-neighbors) to non out-neighbors (in-neighbors) and vice versa. A \emph{Seidel-switch} of a vertex in an undirected graph changes neighbors to nonneighbors and nonneighbors to neighbors.  Performing an out-switch and in-switch on the same vertex of an undirected graph is equivalent to Seidel-switching that vertex.  Let $\neg^+_v(G)$, $\neg^-_v(G)$, and $\neg_v(G)$ be the graphs obtained by out, in, and Seidel switching on a vertex $v \in V(G)$ respectively.  It is easy to see that the order in which vertices are switched does not matter, so let $\neg^+_U(G)$, $\neg^-_U(G)$, and $\neg_U(G)$ be the graph obtained by out, in, and Seidel switching on a set $U \subseteq V(G)$.  If a mixed sequence of in and out switches are permitted, then let $\neg^{\pm}_{I,O}(G)$ be the graph obtained by \emph{Gale-Berlekamp switching} where $I,O \subseteq V(G)$ are the subsets of vertices that have been in-switched and out-switched respectively.

\begin{definition}
Let $G,H \in \mathcal{G}$.  Then $G \sim^* H$ iff $\exists U \subseteq V$ such that $\neg^*_U(G) \cong H$ with respect to some switching operation $*$.
\end{definition}

\begin{proposition}
$\sim^x$ is an equivalence relation over $\mathcal{G}$.
\end{proposition}

\begin{definition}
Let $\mathcal{C}^*_G = \{H \in \mathcal{G} : G \sim^x H\}$.  Then $\mathcal{C}^*_G$ is the switching class of $G$ with respect to some switching operation $*$.
\end{definition}

\noindent 
In particular, we let $\mathcal{C}^{+}_G$, $\mathcal{C}^{-}_G$, $\mathcal{C}^{\pm}_G$, and $\mathcal{C}_G$ denote the \emph{in, out, Gale-Berlekamp, and Seidel switching class of} $G$ respectively.  It is worth noting that in-switching classes and out-switching classes have an algebraic structure similar to Gale-Berlekamp switching classes~\cite{RothV08}. It is routine to show that  $\sim^+$ and $\sim^-$ form an equivalence relation over $\mathcal{G}$ that gives rise to the Abelian group $\mathbb{Z}_2^{n^2-2n}$.

\begin{definition}
The partially complemented adjacency list (pc-list) of a graph $G$ (with respect to some switching operation) is an adjacency list outfitted with a constant number of bitstrings of length $n$ that represent vertex switches.
\end{definition}
If a vertex $v$ is switched, then we let $\widetilde{N}(v)$ denote its doubly-linked neighborlist in the pc-list of $\widetilde{G}$.  If $v$ is unswitched, then we let $N(v)$ denote the doubly-linked neighbor list of $v$ in the pc-list of $\widetilde{G}$.  For any switched vertex $v$, we still refer to $\widetilde{N}(v)$ as the neighborlist of $v$ even though its elements are actually non-neighbors in the original graph.

\begin{proposition}
$\overline{G}$ is the graph obtained by out-switching (in-switching) all of the vertices.
\end{proposition}
\noindent The proposition above is useful due to the fact that some graph classes can be recognized by considering properties of their complements~\cite{McConnellS99}. Unfortunately, constructing the complement graph $\overline{G}$ is an $\Omega(n^2)$ operation which precludes any linear-time bound.  The pc-list has proved useful in this context since it represents $\overline{G}$ implicitly which obviates the $\Omega(n^2)$ cost of constructing $\overline{G}$~\cite{DahlhausGM02,McConnellS99}.

The pc-list was motivated by McConnell's \emph{complement-equivalence classes}~\cite{McConnell97}.  It is straightforward to see that the symmetric complement-equivalence classes of~\cite{DahlhausGM02} coincide with Seidel switching classes and in-out complement-equivalence classes~\cite{DahlhausGM02} coincide with Gale-Berlekamp switching classes.  Due to this correspondence, it seems natural to couch the pc-list in terms of the existing theory of switching classes. 

It is obvious that we should seek out small members of switching classes to obtain a more succinct representation of a given graph.  Ideally, we should seek a member of a switching class with the fewest edges.
\begin{definition}
A minimum representative of a switching class $\mathcal{C}^*_G$ is a not necessarily unique graph $\widetilde{G} \in \mathcal{C}^*_G$ having minimum edge cardinality $\widetilde{m}$.
\end{definition}
If we limit ourselves to strictly out-switches or strictly in-switches, the following lemma shows that we can easily construct a minimum representative using a greedy algorithm.
\begin{lemma}\label{lem:buildOut}
\cite{DahlhausGM02} A minimum representative $\widetilde{G} \in \mathcal{C}^+_G$ ($\widetilde{G} \in\mathcal{C}^-_G$) can be constructed in $O(n+m)$ time.
\end{lemma}
\begin{proof}
Visit each vertex and if switching it reduces the edge count, do so. For out switching, if there are more than $n/2$ elements in $v$'s neighbor list, switch $v$ and replace the neighbor list with non-neighbors of $v$.  The work for creating the list of non-neighbors can be charged to visiting the neighbors of $v$. For in switching, if $v$ appears more than $n/2$ times in the adjacency list of $G$, then switch $v$.  The work is clearly $O(m)$.
\end{proof}

\noindent Observe that if both in-switches and out-switches are allowed, then the algorithm in the proof of Lemma~\ref{lem:buildOut} no longer guarantees that the representative is a minimum.  This is because edges can reappear while constructing the representative.  It is known that computing minimum representatives for Gale-Berlekamp and Seidel switching classes is NP-hard and is even hard to approximate within a constant factor of the optimum~\cite{RothV08,JelinkovaSHK11}.  There do however exist randomized linear-time $(1+\epsilon)$-approximation schemes for computing a minimum representative $\widetilde{G} \in \mathcal{C}^\pm_G$~\cite{KarpinskiS09}.  This allows one to obtain a representative $G' \in \mathcal{C}^\pm_G$ such that $|G'| = \Theta(|\widetilde{G}|)$ in $O(n\log(n)+m\log(n))$ time with high probability.

\section{Basic Algorithms for Switching Classes}

\subsection{Traversal}

Traversal algorithms for out-switching classes first appeared in~\cite{DahlhausGM02} where the existence of $O(n+\widetilde{m})$ algorithms for traversal on Seidel switching classes was left open.  We show that these algorithms can obtained in a straightforward manner through a slight modification of the pc-list data-structure and the traversal algorithms of~\cite{DahlhausGM02,LindzeyO13}.
 We refer the reader to~\cite{DahlhausGM02,LindzeyO13} for a more thorough treatment.  We begin with an intuitive explanation as to why the pc-list is able to provide asymptotic savings in runtime for graph traversal.

Let $\mathcal{A}$ be a graph traversal algorithm and assume there exists an oracle $\mathcal{O}$ such that for any current vertex $v$, it returns in $O(1)$ time either an undiscovered neighbor $v$ or reports that all of $v$'s neighbors have been discovered.  If $\mathcal{A}$ considers a vertex that has already been discovered, then we shall call this a \emph{bad query}. It is clear that the runtime of $\mathcal{A}$ with oracle $\mathcal{O}$ is $\Theta(n)$ since $\mathcal{A}$ can make no bad queries.  This is no longer the case if we run $\mathcal{A}$ without $\mathcal{O}$ since we might have $\omega(n)$ bad queries for arbitrary graphs. In this case, the runtime of $\mathcal{A}$ is dominated by bad queries since we could have as many as $O(m)$.  However, if the size of $G$'s pc-list is asymptotically smaller than its adjacency list, then we can obtain a tighter upper bound on the number of bad queries that can occur during an execution of algorithm $\mathcal{A}$ without use of an oracle.  This is due to the fact that every bad query of BFS or DFS can be charged to an element of the pc-list data-structure~\cite{DahlhausGM02,LindzeyO13}.

\begin{theorem}
\cite{DahlhausGM02} Given $\widetilde{G} \in \mathcal{C}^+_G$, BFS on $G$ can be done $O(n+\widetilde{m})$ time.
\end{theorem}

\begin{theorem}
\cite{DahlhausGM02,LindzeyO13} Given $\widetilde{G} \in \mathcal{C}^+_G$, DFS on $G$ can be done in $O(n+\widetilde{m})$ time.
\end{theorem}
\begin{proposition}
Let $S \subseteq V$ be a set of Seidel switched vertices and let $H = G[S] \cup G[S-V]$.  Then $\neg_S(G)$ is isomorphic to the graph $H' = (V(H),E(H) \cup \overline{E}(S,V-S))$ where $\overline{E}(S,V-S)$ is the complement of the cut induced by $(S,V-S)$.
\end{proposition}
Given an adjacency list representation of $G$ and a set of Seidel switched vertices $S \subseteq V$ such that $|E(\neg_S(G))| < |E(G)|$, a pc-list data-structure that represents $\neg_S(G)$ can be constructed in $O(n+m)$ time as follows.  Let $v$ be an arbitrary vertex. If $v$ is switched, set its bit to 1, add all of its neighbors in $S$ and its nonneighbors in $V-S$ into its neighborlist.  If $v$ is unswitched, set its bit to 0, add all of its neighbors in $V-S$ and its nonneighbors in $S$ into its neighborlist. Relabel the vertices so that the members of $V-S$ have a smaller label than members of $S$, then radix-sort the pc-list with respect to this new labeling.  The sort has the effect of making all nonneighbors of $v$ appear consecutively in $v$'s neighborlist.  Finally, insert a dummy vertex between the two elements $u$ and $w$ of $v$'s neighborlist such that $u$ is switched and $w$ is unswitched.

The same algorithms given in~\cite{DahlhausGM02} and~\cite{LindzeyO13} can be used on this pc-list representation with the following modification; once $v$'s dummy vertex is visited during a scan of its neighbor list, flip $v$'s bit in the pc-list.  If $v$ is switched, then the flip has the effect of treating elements after the dummy vertex as actual neighbors of $v$. If $v$ is unswitched, then the flip has the effect of treating elements after the dummy vertex as nonneighbors of $v$.  Accounting for this modification in the traversal algorithms of~\cite{DahlhausGM02,LindzeyO13} is trivial.  The foregoing gives the following result.
\begin{theorem}
Given $\widetilde{G} \in \mathcal{C}_G$, BFS and DFS on $G$ can be done in $O(n+\widetilde{m})$ time.
\end{theorem}

\subsection{Contraction}

\noindent Another basic graph-theoretic operation is contraction of a subset of vertices.  Henceforth we shall assume that all switching operations are out-switches for ease of exposition.  Let $\widetilde{G} \in \mathcal{C}_G^+$ be a minimum representative, let $n(B)$ be the number of vertices of a subset $B \subseteq V$, $\widetilde{m}(B)$ be the number of edges incident to vertices of $B$ in the graph $\widetilde{G}$.  The results of this section will be needed for computing maximum matchings of general graphs.  Without loss of generality, we assume that the pc-list of $\widetilde{G}$ is sorted by vertex label.

\begin{lemma}\label{lem:contract}
Given $\widetilde{G} \in \mathcal{C}_G^+$, a set $B \subseteq V$ can be contracted to vertex $\hat{b}$ in $O(n(B)+\widetilde{m}(B))$ time.
\end{lemma}
\begin{proof}
To build the neighbor list of $\hat{b}$, we first build two doubly-linked neighbor lists $X,Y$ that correspond to the contraction of all the switched vertices $S \subseteq B$ and unswitched vertices of $B-S$ respectively.

Initialize $X$ to be $\widetilde{N}(s)$ some $s \in S$ and let $b,c \in S$.  Then contracting $b$ and $c$ together corresponds to taking the intersection $\widetilde{N}(b) \cap \widetilde{N}(c)$.  Since the neighborlists are sorted, taking the intersection $\bigcap_{b \in S} \widetilde{N}(b)$ can be computed in $O(n(B)+\widetilde{m}(B))$ using a routine similar to the merge routine of merge-sort as follows.  Let $L[i]$ denote the $i$th element of a doubly-linked list $L$.

If $X[i] = \widetilde{N}(b)[j]$, then the comparison can be charged to the $j$th edge of $b$'s neighbor list.  If $X[i] > \widetilde{N}(b)[j]$, then the comparison can be charged to the $j$th edge of $b$'s neighbor list.  If $X[i] < \widetilde{N}(b)[j]$, then $X[i]$ is removed from the doubly-linked list $X$.  Removing a vertex from $X$ happens at most $\widetilde{m}(B)$ times since each deletion can be charged to an element of $\widetilde{N}(s)$.  

Initialize $Y$ to be $N(v)$ for some $v \in B-S$ and let $b,c \in B-S$. Suppose that $b,c \in B-S$ are unswitched, then contracting $b$ and $c$ together corresponds taking the union $N(b) \cup N(c)$.  The union $Y = \bigcup_{b \in B-S} N(b)$ can clearly be computed in $O(n(B)+\widetilde{m}(B))$.

To combine $X$ and $Y$ it suffices to scan $X$ and remove all elements from $X$ that also exist in $Y$. This can be done using a routine similar to the merge routine.  At the end of the routine, define $\widetilde{N}(\hat{b})$ to be $X$.  Since the sizes of $X$ and $Y$ are each $O(n(B) + \widetilde{m}(B))$, it follows that $\widetilde{N}(\hat{b})$ can be constructed in $O(n(B) + \widetilde{m}(B))$ time.
\end{proof}

\begin{lemma}\label{lem:contractG}
Let $\beta$ be a collection of vertex-disjoint subsets. Then the contracted graph $\widetilde{G}/\beta$ can be computed in $O(n+\widetilde{m})$.
\end{lemma}
\begin{proof}
By Lemma~\ref{lem:contract} we can perform all of the contractions in time $\sum_{B\in\beta} n(B) + \widetilde{m}(B) = O(n + \widetilde{m})$.  After performing the contractions, the pc-list must be cleaned up so that vertices subsumed by contractions are no longer referenced.  This can be done by radix-sorting and removing duplicates in $O(n + \widetilde{m})$ time.
\end{proof}

\section{Super-Linear Graph Algorithms}
In this section, we show that given a graph $G = (V,E)$, spending $O(n+m)$ time to compute an $O(n+\widetilde{m})$ space pc-list representation of $G$ gives rise to better bounds for several canonical unweighted graph algorithms.  

\subsection{Diameter and Transitive Closure}

At present, the most efficient combinatorial algorithm (ignoring log factors) for computing the diameter and transitive closure of a graph to our knowledge is the naive $O(n^2 + nm)$ algorithm, that is, calling BFS from each vertex.  The following results are straightforward.
\begin{theorem}
The diameter of a graph $G$ can be computed in $O(n^2+ n\widetilde{m})$ time.
\end{theorem}
\begin{theorem}
The transitive closure of a graph $G$ can be computed in $O(n^2+ n\widetilde{m})$ time.
\end{theorem}
\begin{proof}
Compute $\widetilde{G}$ in $O(n+m)$ time,  then run the naive algorithm from each vertex using the BFS algorithm of~\cite{DahlhausGM02}.  The runtime of this algorithm is $O((n+m) + n(n+\widetilde{m})) = O(n^2 + n\widetilde{m})$.
\end{proof}
\noindent Since $n^2+n\widetilde{m}$ is never worse than $n^2+nm$ and is sometimes better, this is indeed a better bound for both diameter and transitive closure.  

At this point it is natural to ask when a graph \emph{benefits} from its pc-list representation, that is, when its pc-list representation is asymptotically smaller than its adjacency-list representation.  It is easy to see that there are at least twice as many graphs that benefit from pc-lists as there are graphs that benefit from adjacency-lists.  This is because the pc-list is a generalization of the adjacency-list and the complement of any graph whose adjacency-list representation is asymptotically smaller than its adjacency-matrix representation must have a minimum representative $\widetilde{G} \in \mathcal{C}_G^+$ such that $\widetilde{m} = o(m)$.  For instance, the class of graphs whose complement is sparse benefits from its pc-list representation simply by out-switching all of its vertices. 

It would be interesting to give precise conditions for when a graph benefits from its pc-list representation, but for now, our intuition tells us that very dense graphs and ``unbalanced" graphs (graphs with vanishingly few vertices of average valency) appear to benefit most from the pc-list representation.\footnote{Perhaps it is possible to make this intuition well-defined via discrepancy theory.}  On the other hand, since almost all graphs are roughly $n/2$-regular, the pc-list does not provide an asymptotically smaller representation for most graphs.  The techniques of~\cite{CheriyanM96} and~\cite{FederM95} are notable since they give logarithmic speedups in runtime for almost all graphs.

\subsection{Hopcroft-Karp Bipartite Maximum Matching}

Notice that switching in general does not preserve the bipartite property.  Given a bipartite graph $G$, define the \emph{bipartite switching class} of $G$ (with respect to some switching operator), such that neighbors of $v$ become nonneighbors and nonneighbors of $v$ that lie outside of $v$'s partition class become neighbors.  In this section, all switches are assumed to be bipartite out-switches.  Familiarity with the bipartite maximum matching problem is assumed.

The Hopcroft-Karp algorithm consists of \emph{phases}, each of which strictly increases the size of a current matching.  In~\cite{HopcroftK73} it is shown that a phase consists of a call to a modified BFS routine followed by a call to a modified DFS and that only $O(\sqrt{n})$ phases are needed to compute a maximum matching.  These routines can be implemented to run in $O(n+m)$ time which gives rise to a $O(\sqrt{n}(n+m))$ bound for computing a maximum matching of a bipartite graph.  

The purpose of running the modified BFS is to discover a directed acyclic level graph $L$ such that any path connecting two unmatched vertices in $L$ corresponds to a shortest augmenting path in $G$.  A modified DFS is then conducted on $L$ in order to find a maximal set of vertex-disjoint augmenting paths $\mathcal{P}$ in $G$. These steps are repeated until a phase is encountered such that $\mathcal{P} = \emptyset$, in which case the matching $M$ must be maximum by a theorem of Berge. 

It suffices to show that given $\widetilde{G}$, a phase of Hopcroft-Karp can be implemented to run in $O(n+\widetilde{m})$ time. Let BFS* and DFS* be pc-list implementations of the aforementioned modified BFS and DFS routines.  The BFS* routine is essentially the same as the BFS algorithm of~\cite{DahlhausGM02} except for the following modifications.  Let $level(v)$ denote the BFS level of a vertex $v$.
\begin{enumerate}
\item The undiscovered vertices are divided into two doubly-linked lists $U_A$ and $U_B$.  If the current vertex $v \in A$, then BFS$^*$ only considers vertices in $U_B$ (similarly for $v \in B$).
\item The discovered vertices are kept in an array of doubly-linked lists $\mathcal{L}$ that represents a partition of the discovered vertices by BFS level.  In particular, a discovered vertex $v$ resides in the doubly-linked list at index $level(v)$ of $\mathcal{L}$.  This array represents the levels in the DAG $L$.
\item Let $k$ be the level of the first unmatched vertex in $B$ that has been dequeued. Once all vertices at level $k$ have been dequeued, the routine returns $\mathcal{L}$.
\end{enumerate}
 It is clear that the aforementioned modifications to the BFS routine of~\cite{DahlhausGM02} can be accomplished in $O(n+\widetilde{m})$ time.  Pseudocode of BFS* that achieves this bound is given in the appendix.
 
The vertices of $\mathcal{L}$ form the initial set of undiscovered vertices for the modified DFS routine.  Notice that explicitly constructing the level DAG $L$ from $\mathcal{L}$ would exceed the $O(n+\widetilde{m})$ time bound; however, this is unnecessary since $\mathcal{L}$ provides an implicit representation of $L$.
A precondition to the DFS algorithm of~\cite{LindzeyO13} is that doubly-linked list of undiscovered vertices is ordered according to the ordering of the vertices of the pc-list. For each $i \in \{1 \cdots k\}$, the ordering of $\mathcal{L}[i]$ might not respect the ordering of the vertices in the sorted pc-list; however, this can be corrected by performing a radix-sort on $\mathcal{L}$ in $O(n)$ time.  The DFS* routine is the same as the DFS algorithm in~\cite{LindzeyO13} except for the following modifications.

\begin{enumerate}
\item If a vertex $v \in \mathcal{L}[level(v)]$ is current, then $U$ points to the doubly-linked list $\mathcal{L}[level(v)+1]$ so that $v$ only considers undiscovered neighbors at $level(v)+1$.
\item When a path $P$ from the designated source vertex $s$ to an unmatched vertex in $B$ has been found, the routine adds $P-s$ to the set of vertex-disjoint augmenting paths and restarts DFS* from $s$.
\end{enumerate}
The vertices of $P-s$ cannot be considered in subsequent DFS* searches since they are removed from $\mathcal{L}$ once they are discovered.   Since the union of vertex-disjoint paths $\mathcal{P}$ has size $O(n)$, it is clear that DFS* can be implemented to run in $O(n+\widetilde{m})$ time.   Pseudocode of DFS* that achieves this bound is given in the appendix.

Finally, since a set of matched edges $M$ has size $O(n)$, it is clear that the symmetric difference between $M$ and the edges of the paths in $\mathcal{P}$ can be computed in $O(n)$ time.  This completes the description of a phase.  Because BFS* and DFS* are $O(n+\widetilde{m})$ and updating a matching takes $O(n)$ time, we obtain the following result.

\begin{lemma}\label{lem:phase}
A phase of Hopcroft-Karp can be implemented in $O(n+\widetilde{m})$ time.
\end{lemma}

\noindent By Lemma~\ref{lem:phase} and the fact that only $O(\sqrt{n})$ phases are needed to compute a maximum matching of a bipartite graph,  we obtain the following result.
\begin{theorem}
A maximum matching of a bipartite graph $G$ can be computed in $O(n^{1.5}+\sqrt{n}\widetilde{m} +m)$ time.
\end{theorem}

\noindent Pseudocode for Hopcroft-Karp that achieves the bound above is given in the appendix.  In~\cite{BalasN93} it is shown that finding maximum matchings of complement biclique graphs can be used as a heuristic to enumerate large cliques in graphs.  It is possible that the pc-list could make this heuristic more efficient in practice by circumventing the construction of the complement.

\subsection{Gabow-Tarjan Maximum Matching}

A full discussion of the Gabow-Tarjan maximum cardinality matching algorithm~\cite{GabowT91} is beyond the scope of this paper.  The following is only a rough sketch of the algorithm.  Familiarity with the maximum matching problem is assumed.\\

\noindent It is well known that the maximum matching problem for general graphs is complicated by the existence of alternating odd cycles (blossoms)~\cite{Edmonds65}.  
The Gabow-Tarjan maximum matching algorithm consists of $\lceil \sqrt{n} \rceil$ iterations of so-called \emph{phase 1} (not to be confused with Hopcroft-Karp's phases) followed by $O(\sqrt{n})$ calls to \emph{find\_ap\_set}~\cite{GabowT91}.  
Phase 1 consists of running \emph{find\_ap\_set} on a contracted graph $G/\beta$ and expanding unweighted blossoms afterwards. 
The \emph{find\_ap\_set} routine is a depth-first based search responsible for discovering augmenting paths and blossoms of the input graph. A pre-condition to \emph{find\_ap\_set} is that the input graph $G/\beta$ is contracted with respect to a set of blossoms $\beta$.  Once \emph{find\_ap\_set} halts, the routine returns a maximal set of vertex-disjoint augment paths in the contracted graph.  These paths in the contracted graph are translated into vertex-disjoint augmenting paths in the original graph in $O(n)$ time and the current matching is updated with respect to those paths.  In addition to finding these paths, an execution discovers a set of new blossoms as well as a set of previously discovered blossoms to be expanded. This gives rise to a new set of blossoms to be contracted before the next call to \emph{find\_ap\_set}.  A blossom is expanded if its dual variable $z$ becomes non-positive, which establishes an invariant that contracted blossoms $\hat{b}$ always have positive weight after the first execution of \emph{find\_ap\_set}.  Unweighted blossoms are expanded because they might be ``hiding augmenting paths"~\cite{Edmonds65}.  Expanding these blossoms gives \emph{find\_ap\_set} a chance to find these hidden augmenting paths in the next iteration. The maximum cardinality case is treated as a special case of their minimum weight matching algorithm and is shown to run in time $O(\sqrt{n}(n+m))$ since the algorithm performs $O(\sqrt{n})$ iterations of \emph{find\_ap\_set} which runs in time $O(n+m)$. 
In~\cite{GabowT91} it is shown that blossom, augmenting path, and dual variable maintenance all take $O(n)$ time and space during an iteration of phase 1.  

Lemma~\ref{lem:contractG} shows that we can build the contracted graph in time proportional to $\widetilde{G}$, so it remains to show that \emph{find\_ap\_set} can be conducted in $O(n+\widetilde{m})$ time.  We shall assume that we have obtained a minimum member $\widetilde{G}$ of $G$'s out-switching class.  Recall that a vertex of $\widetilde{G}$ is switched if and only if it has $n/2$ or more neighbors in the original graph.  It follows that we can build an adjacency lookup vector $\mathbf{v}[]$ for each switched vertex $v$ in $O(n+m)$ time.  This lookup vector will allow us to answer if some vertex $u$ is adjacent to a switched vertex $v$ in the original graph in $O(1)$ time. 

\begin{lemma}
The find\_ap\_set routine can be implemented in $O(n + \widetilde{m})$ time.
\end{lemma}

\noindent It suffices to show that the subroutine \emph{find\_ap} can be implemented in $O(n+\widetilde{m})$ time~\cite{GabowT91}.  Let $b(v)$ denote the blossom that $v$ currently belongs to.  It is important to note that \emph{find\_ap} only performs grow steps and blossom steps.  We show how to perform each of these steps when the current vertex is switched.  The following invariants will be needed to simplify the analysis of our modification to \emph{find\_ap}.

\begin{enumerate}
\item Let $x$ be the current vertex. If an edge $xy$ is scanned and $b(y)$ became outer after $b(x)$, then a blossom step is performed and every vertex in $b(y)$ has been completely scanned~\cite{GabowT91}.
\item The order in which vertices become outer is given by the ordering of a doubly-linked list $OUT$.
\item If the current vertex performs a blossom step, then no more grow steps are possible from the current vertex.
\end{enumerate}
The second and third invariants are not part of the specification of their algorithm; however, it can be implemented to maintain these invariants without affecting the correctness or resource bounds of the algorithm. 

It is straightforward to modify the DFS algorithm of~\cite{LindzeyO13} so that when a current outer vertex $u$ is considering an inner undiscovered neighbor $v$, the next outer current vertex becomes $v'$ where $vv' \in M$ is a matched edge.  The details of the blossom step are slightly more involved.  It helps to view the blossom steps as DFS where the ``undiscovered vertices" are those vertices that have an outer label that is greater than the current vertex's outer label.  Once a blossom step is performed, the inner vertices along that blossom become current (and therefore outer) in order of their DFS discovery time and the routine proceeds recursively.


Assume that a current vertex $x$ has no more grow steps, that is, there are no more undiscovered vertices reachable from $x$.
If $x$ is not switched, then proceed as usual in \emph{find\_ap}; otherwise, assume that $x$ is switched.
The invariants guarantee that the blossom steps from $x$ can be conducted by performing DFS where the ``undiscovered vertices" are those vertices to the right of $x$ in $OUT$.  Let us refer to this ordered set of vertices as $OUT_{>x}$.  The main obstacle is that we cannot use the DFS routine of~\cite{LindzeyO13} immediately to solve this problem since the ordering of the pc-list does not respect the ordering of $OUT_{>x}$.  For this reason we must use $\mathbf{x}[]$ to determine adjacency in $O(1)$ time.  

Let $y \in OUT_{>x}$.  If $\mathbf{x}[y] = 1$, then a blossom step is performed. When a blossom step is performed, then $b(y)$ is removed from $OUT$ since invariant 1 implies that these vertices have already been completely scanned and cannot be further used to discover an augmenting path.  Assuming the blossom data-structures in~\cite{GabowT91}, it follows that blossom steps take $O(n+\widetilde{m})$ time.  

If $\mathbf{x}[y] = 0$, then no blossom step is performed, but we can charge the lookup to $y$'s entry in $\widetilde{N}(x)$.  But as in the routines of~\cite{DahlhausGM02} and~\cite{LindzeyO13}, we must guarantee that $x$ considers $y$ only $O(1)$ times throughout the entire DFS execution on $OUT_{>x}$.  To ensure this, each time we encounter a $y$ such that $\mathbf{x}[y] = 0$, we add it to the end of an auxiliary doubly-linked neighborlist for $\widetilde{N}'(x)$.  This has the effect of building an ordered neighborlist for $x$ that respects the ordering of $OUT$.  Using this neighborlist, we can follow the restart step of the DFS algorithm in~\cite{LindzeyO13} to correctly and efficiently resume DFS on $OUT_{>x}$ when $x$ returns from a recursive call. As in the bipartite case, keeping track of augmenting paths and updating the current matching can be accomplished in $O(n)$ time.  Since \emph{find\_ap\_set} can be implemented in $O(n + \widetilde{m})$ time, we obtain the following result.

\begin{theorem}\label{thm:maxmatch}
A maximum matching of a graph $G$ can be computed in $O(n^{1.5}+\sqrt{n}\widetilde{m} +m)$ time.
\end{theorem}

\section{Conclusions}

We have demonstrated that switching classes can be used to obtain asymptotically better bounds for several graph algorithms through use of the pc-list data-structure.  These improvements on algorithm resource bounds suggest that the pc-list is a more efficient data-structure than the adjacency list for several unweighted graph problems.  But like any graph representation, it has its trade-offs.  For instance, finding an arbitrary neighbor of a switched vertex $v$ in the original graph takes $\Theta(|\widetilde{N}(v)|)$ time whereas finding an arbitrary neighbor of an unswitched vertex takes $\Theta(1)$ time.

It may be tempting to believe that any unweighted graph algorithm can be implemented to work with a pc-list representation; however, it seems unlikely that the maximum matching algorithm of~\cite{MicaliV80} is amenable to pc-lists.  This is due to the fact that their approach requires a number of edges (so-called bridges) to be queued and processed at a later point in the execution of the algorithm.  When a bridge is processed, it does not always progress the algorithm, that is, it may not lead to a grow step, produce an augmenting path, or discover a blossom.  In light of this, processing a bridge cannot always be charged to a vertex or edge of $\widetilde{G}$.  We were unable to obtain a bound tighter than $O(m)$ for the number of bridges queued throughout the algorithm; however, it might be possible to modify the algorithm so that only bridges that lead to an augmenting path or blossom are considered.

In fact, there are classical super-linear graph problems that inherently cannot benefit from the pc-list.  Two such examples are finding an Eulerian tour of $G$ and computing the ear decomposition of $G$.  Any algorithm for these problem must spend $\Omega(m)$ time for arbitrary graphs since the size of the output is proportional to the number of edges.

An obvious line of future work would be towards a more formal characterization of graphs that benefit from pc-list representations as well as the development of more pc-list algorithms.  We conclude by thanking Ross M. McConnell for his insightful comments.



\bibliographystyle{plain}
\bibliography{master}

\section{Appendix}

\begin{figure}[htp]\label{fig:pclist}

\includegraphics[scale=.14]{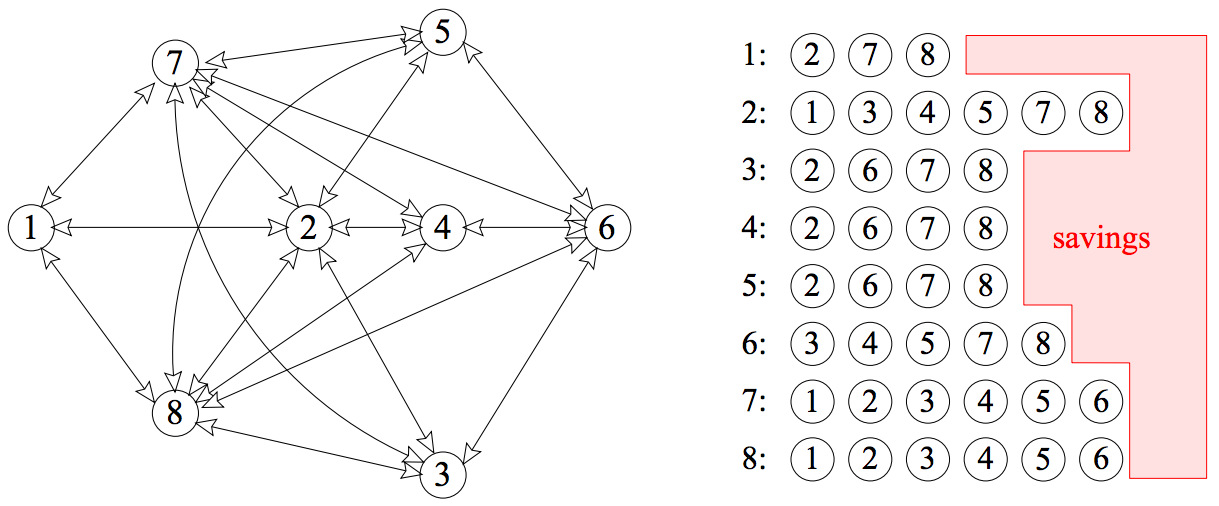} \includegraphics[scale=.14]{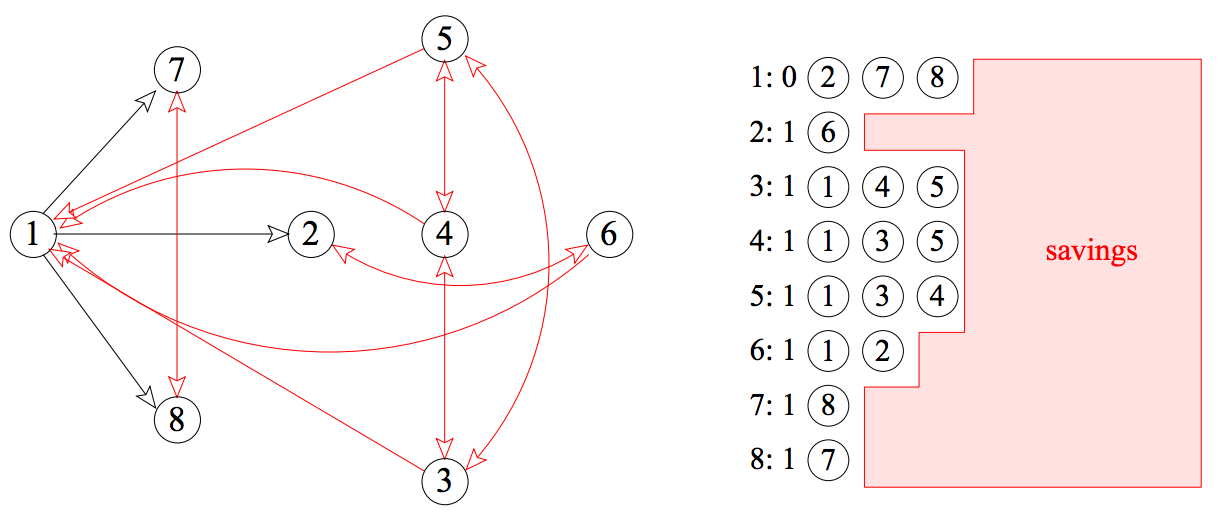}
\centering
\includegraphics[scale=.14]{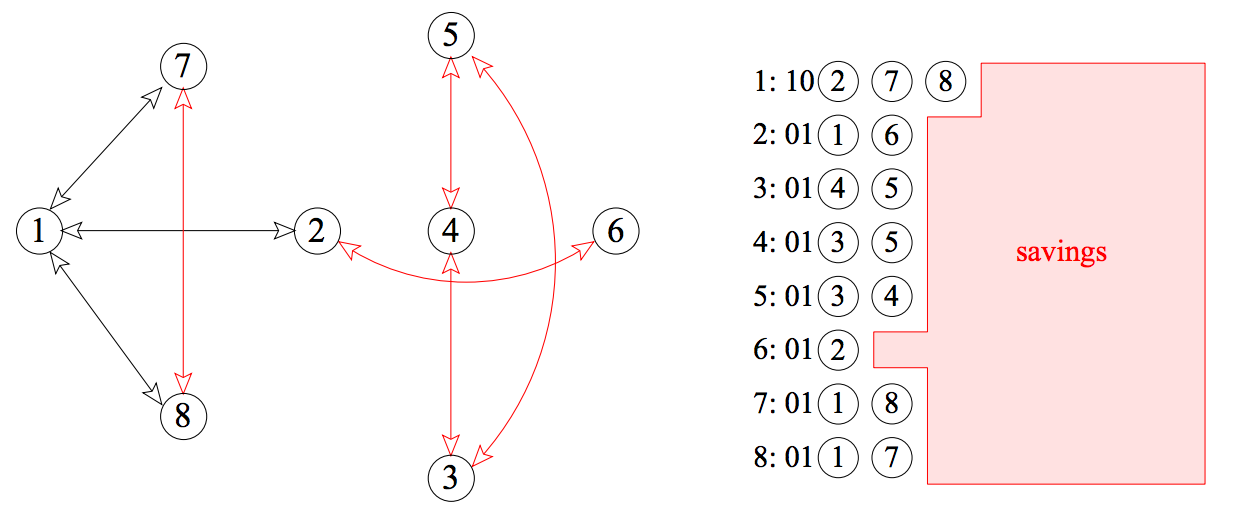}
\caption{An illustration of the partially complemented adjacency list.  Setting $G$ as the graph given in the leftmost picture, it is not hard to see that the graph obtained by out switching (the rightmost picture) and Gale-Berlekamp switching (the lower picture) are minimum representatives of $\mathcal{C}_G^+$ and $\mathcal{C}_G^{\pm}$ respectively.  This demonstrates that $\overline{G}$ is not necessarily a minimum representative of $G$.}
\end{figure}

\begin{algorithm}
\KwData{A bipartite graph $G = (A,B,E)$}
\KwResult{A maximum matching $M$ of $G$}
$M \longleftarrow \emptyset$\;
Let $\widetilde{G}'$ be the pc-list of a minimum representative of $\mathcal{C}^+_G$\;
\Repeat{$\mathcal{P} = \emptyset$}{
Let $\widetilde{G}$ be a copy of the pc-list $\widetilde{G}'$\;
$\mathcal{P} \longleftarrow \emptyset$\;
$U_A,U_B \longleftarrow A,B$\;
connect $s$ to each unmatched vertex in $A$\;
$\mathcal{L} \longleftarrow$ BFS$^*(s)$\;
$\mathcal{L} \longleftarrow$ radix-sort($\mathcal{L}$)\;
DFS*($s$)\;
$M \longleftarrow \mathcal{P} \bigoplus M$\;
}
return $M$\;
\caption{Hopcroft-Karp($G$)}  \label{HK}
\end{algorithm}

\begin{algorithm} \label{alg:BFS}
\KwData{A pc-list $\widetilde{G}$, two global ordered doubly-linked lists $U_A$, $U_B$ of the undiscovered vertices of $A$ and $B$ respectively.}
\KwResult{A partition $\mathcal{L}$ that represents a level graph.}
Let $\mathcal{L}$ be an array of doubly-linked lists\;
Let $Q$ be a queue\;
$k \longleftarrow +\infty$\;
\texttt{enqueue($s$,$Q$)}\;
\Repeat{$Q$ is empty}{
$v \longleftarrow$ \texttt{dequeue($Q$)}\;

\If{$level(v) > k$}{
 break\;
}  

\ElseIf{$v \in A$}{
Let $U$ point to $U_B$\;
}
\Else{
\If{$v$ is unmatched}{
$k \longleftarrow level(v)$\;
}
Let $U$ point to $U_A$\;
}

\If{$v$ is unswitched}{
  \For{$u \in N(v)$}{
     \If{$u$ is undiscovered}{
     	\texttt{remove($u$,$U$)}\;
     	\texttt{enqueue($u$,$Q$)}\;
     	\texttt{append($u$,$\mathcal{L}[level(v)+1]$)}\;
     }
  }
}
\Else{

\For{$u \in \widetilde{N}(v)$}{
   \texttt{mark($u$,$U$)}\;
   }
 \For{$u \in U$}{
     \If{$u$ is unmarked}{
     	\texttt{remove($u$,$U$)}\;
     	\texttt{enqueue($u$,$Q$)}\;
     	\texttt{append($u$,$\mathcal{L}[level(v)+1]$)}\;
     }
  }

}
}
return $\mathcal{L}$
\caption{BFS$^*$($s$)} 
\end{algorithm}

\begin{algorithm} \label{alg:DFS}
\caption{DFS*($v$)} 
\KwData{A pc-list $\widetilde{G}$, a current undiscovered vertex $v$, a global array $\mathcal{L}$ of ordered doubly-linked lists of undiscovered vertices, and global set of vertex-disjoint shortest augmenting paths $\mathcal{P}$.}
\KwResult{Discovers a maximal set of vertex-disjoint augmenting paths $\mathcal{P}$.}
Let $U = \mathcal{L}[level(v)+1]$\;
\If{$v \neq s$}{
\texttt{remove}($\mathcal{L}[level(v)]$,$v$)\;
}
\If{$v \in B$ and $v$ is unmatched}{
Let $P$ be vertices of the DFS stack\;
$\mathcal{P} \longleftarrow \mathcal{P} \cup \{P - s\}$\;
DFS*($s$)\;
}
\If{$v$ is unswitched}{
	\For{$u \in N(v)$}{
	 	   DFS*($u$)\;
	}
}
\Else{
	$u_v \longleftarrow$ \texttt{head}($U$)\;
	$n_v \longleftarrow$ \texttt{head}($\widetilde{N}(v)$)\;
	\While{$u_v \neq null$}{
	\If{$u_v = n_v$}{
		$u_v \longleftarrow$ \texttt{next}($U,u_v$)\;
		$n_v \longleftarrow$ \texttt{next}($\widetilde{N}(v),n_v$)\;
	}
	\ElseIf{$u_v > n_v$}{
		$s \longleftarrow n_v$\;
		$n_v \longleftarrow$ \texttt{next}($\widetilde{N}(v),n_v$)\;
		\texttt{remove}($\widetilde{N}(v),s$)\;
	}
\Else{
	DFS($u_v$)\;
	\tcp{restarting step ...}
        $w$ = prev($\widetilde{N}(v),n_v$)\;
	\While{$w \neq null$ and $w \notin U$}{
		$t = w$\;
		$w \longleftarrow$ \texttt{prev}($\widetilde{N}(v),t$)\;	
		\texttt{remove}($\widetilde{N}(v), t$))\;
	}
	\If{$w = null$}{
		$u_v \longleftarrow$ \texttt{head}($U$)\;
	}
	\Else{
		$u_v \longleftarrow$ \texttt{next}($U,w$)\;
	}
}
}
}

\end{algorithm}

\begin{algorithm}\label{alg:contract}

\For{$B_i \in \beta$}{
$X, Y \longleftarrow \emptyset$\;
\ForEach{$b \in B_i = \{b_0,b_1,\cdots,b_k\}$}{
\If{b is switched}{
	\If{$X = \emptyset$}{ $X \longleftarrow \widetilde{N}(b)$\;}
	\Else{$X \longleftarrow X \cap \widetilde{N}(b)$\;}
}
\Else{
	$Y \longleftarrow Y \cup N(b)$\;
 }
}

$\widetilde{N}(\hat{b}_i) \longleftarrow X \bigoplus Y$\;
}
Relabel each vertex $v \in B_i$ with $\hat{b}_i$\;
$\widetilde{G} \longleftarrow$ radix-sort($\widetilde{G}$)\;
$\widetilde{G} \longleftarrow$ remove-duplicates($\widetilde{G}$)\;
\caption{contract($\beta$)}
\end{algorithm}

\begin{algorithm}
Let $\widetilde{G}$ be the pc-list of a minimum representative of $\mathcal{C}^+_G$\;
Let $T$ be the blossom tree of $\widetilde{G}$\;
$i \longleftarrow 0$\;
$\beta \longleftarrow \emptyset$\;
$\widetilde{G}'\longleftarrow \widetilde{G}$\;

\Repeat{$i > \lceil \sqrt{n} \rceil$}{ \tcp{phase 1}
$\widetilde{G} \longleftarrow$ contract($\beta,\widetilde{G}'$)\;
$\mathcal{P},T \longleftarrow$ \emph{find\_ap\_set}($\widetilde{G}$)\;
\tcp{search performs dual adjustments \& expands nonpositve blossoms}
$\beta \longleftarrow$ \emph{search}($\widetilde{G},T$)\;
$M \longleftarrow \mathcal{P} \bigoplus M$\;
$i \longleftarrow i + 1$\;
}
\Repeat{$\mathcal{P} = \emptyset$}{
$\mathcal{P} \longleftarrow$  \emph{find\_ap\_set}($\widetilde{G}$))\;
$M \longleftarrow M \bigoplus \mathcal{P}$\;
}
return $M$\;

\caption{Gabow-Tarjan($G$)}
\end{algorithm}

\end{document}